\documentclass[format=acmsmall, review=false]{acmart}

\usepackage{acm-ec-20}

\usepackage{amsthm}
\usepackage{amsmath}
\usepackage{algorithm}
\usepackage[noend]{algpseudocode}

\newtheorem{theorem}{Theorem}
\newtheorem{lemma}[theorem]{Lemma}
\newtheorem{proposition}[theorem]{Proposition}

\usepackage{footmisc}

\usepackage[most]{tcolorbox}

\newcommand*{\draft}{}

\ifdefined\draft
    \newcommand{\patrick}[1]{\textcolor{blue}{[Patrick: #1]}}
    \newcommand{\andy}[1]{\textcolor{red}{[Andy: #1]}}
    \newcommand{\brendan}[1]{\textcolor{orange}{[Brendan: #1]}}
\else
    \newcommand{\patrick}[1]{}
    \newcommand{\andy}[1]{}
    \newcommand{\brendan}[1]{}
\fi

\title{Reaching Univalency with Subquadratic Communication}

\author{Andrew Lewis-Pye}

\begin{abstract}

The Dolev-Reischuk lower bound establishes that any deterministic Byzantine Agreement (BA) protocol for $n$ processors tolerating $f$ faults requires $\Omega(f^2+n)$ messages. But what exactly does this quadratic cost pay for? Even the minimal requirement that every correct processor \emph{receive at least one message} already necessitates $\Omega(f^2 + n)$ messages. This raises a fundamental question: is the Dolev-Reischuk bound about the difficulty of \emph{reaching univalency}---the point at which the protocol's outcome is determined---or merely about \emph{disseminating} the outcome to all processors afterward?

We resolve this question by showing that reaching univalency does \emph{not} require quadratic communication. Specifically, we introduce $\epsilon$-BA, a relaxation allowing an $\epsilon$-fraction of correct processors to output incorrectly, and prove it can be solved deterministically with $O(n \log n)$ communication complexity when $f < n(1/3 - \epsilon)$. Crucially, any $\epsilon$-BA protocol can serve as the first phase of a full BA protocol: after $\epsilon$-BA, a single all-to-all exchange and majority vote completes BA. Since the outcome is already determined after $\epsilon$-BA, this demonstrates that the quadratic cost in Dolev-Reischuk stems entirely from dissemination, rather than from reaching univalency. We also define Extractable BA for authenticated settings, capturing when processors collectively hold enough signed messages to determine the agreed value, and show it can be solved with communication complexity $O(f \log f)$.
\end{abstract}

\begin{document}

\maketitle

\section{Introduction}

Distributed systems rely on a fundamental primitive: getting processors to agree. Whether coordinating a distributed database, replicating a state machine, or finalising blocks in a blockchain, processors must reach consensus on a common value despite failures. Byzantine Agreement (BA)~\cite{lamport1982byzantine} formalises this problem in its most adversarial form: processors must agree even when some subset of them---the ``Byzantine'' processors---may behave arbitrarily, sending conflicting messages or colluding to disrupt the protocol.

A natural question is: how much communication is required to reach Byzantine Agreement? The Dolev-Reischuk lower bound~\cite{dolev1985bounds} provides a decisive answer: any deterministic protocol  for $n$ processors  tolerating $f$ Byzantine faults must exchange $\Omega(f^2+n)$ messages. This result, now four decades old, remains among the most fundamental in distributed computing. Modern protocols like PBFT~\cite{castro1999practical} and HotStuff~\cite{yin2019hotstuff} are routinely evaluated against it, and considerable effort has gone into achieving this quadratic bound or circumventing it through randomisation~\cite{king2011breaking}.

But what, precisely, does this quadratic cost pay for?

\subsection{Dissemination vs.\ Univalency}

To understand the Dolev-Reischuk bound, it helps to decompose Byzantine Agreement into two conceptually distinct phases:
\begin{enumerate}
    \item \textbf{Reaching univalency:} The point at which the protocol's outcome becomes determined, even if no processor yet knows what it is.
    \item \textbf{Dissemination:} Ensuring each correct processor receives enough information to output the determined value.
\end{enumerate}

A trivial adaptation of the Dolev-Reischuk argument shows that \emph{any} deterministic protocol ensuring every correct processor receives at least one message---regardless of what processors do with those messages---already requires $\Omega(f^2 + n)$ messages.\footnote{For completeness, we include the proof in Appendix \ref{AppA}. However, the proof is essentially identical to the original proof by  Dolev and Reischuk, together with the observation that the same proof also suffices to give the modified claim.} This suggests that much of the quadratic cost may be attributed to dissemination rather than reaching univalency.

This observation raises a natural question:

\begin{quote}
\emph{Is the Dolev-Reischuk lower bound fundamentally about the difficulty of reaching univalency, or is it primarily about the cost of disseminating the outcome once univalency is reached?}
\end{quote}

\subsection{Our Contribution: Univalency is Cheap}

We resolve this question definitively: \textbf{reaching univalency does not require quadratic communication}. The quadratic cost in Dolev-Reischuk is entirely attributable to dissemination. To establish this, we introduce relaxations of Byzantine Agreement that can serve as the first phase of a two-phase BA protocol: the first phase reaches univalency with subquadratic communication, and the second phase handles dissemination. Formal definitions of $\epsilon$-Byzantine Agreement and Extractable Byzantine Agreement are given in Section \ref{setup}. 

\vspace{0.2cm} 
\noindent \textbf{$\epsilon$-Byzantine Agreement.} We allow an $\epsilon$-fraction of correct processors to output incorrectly. Crucially, any protocol solving $\epsilon$-BA can be used as the first phase of a BA protocol: after running $\epsilon$-BA, processors send their outputs to all others and decide by majority vote. Since all but $\epsilon n$ correct processors output the same value $v$  in $\epsilon$-BA, this majority vote yields $v$ for all correct processors, completing BA. We prove:

\begin{theorem}\label{thm:epsilon-ba-intro}
Consider the unauthenticated and synchronous setting. For any constant $\epsilon \in (0,1/3)$, $\epsilon$-BA can be solved deterministically with $O(n \log n)$ communication complexity when $f < n(1/3 - \epsilon)$.
\end{theorem}

\vspace{0.2cm} 
\noindent \textbf{Extractable Byzantine Agreement.} In authenticated settings, i.e., when signatures are available, we define a variant where correct processors must collectively hold enough signed messages to determine the outcome, even if individual processors cannot extract it without further communication. Again, we will see that Extractable BA can serve as the first phase of a BA protocol that suffices to reach univalency. We prove:

\begin{theorem}\label{thm:extractable-ba-intro}
Consider the authenticated and synchronous setting and suppose  $\epsilon\in (0,1/3)$. If $f < n(1/3 - \epsilon)$, Extractable BA can be solved deterministically with $O(f \log f)$ communication complexity.
\end{theorem}

Since both $\epsilon$-BA and Extractable BA serve as the univalency phase of BA, our results demonstrate that reaching univalency is achievable with subquadratic---indeed, nearly linear---communication. The full cost of BA is thus the sum of two components: $O(f \log f)$ to reach univalency, plus a quadratic cost for dissemination.

\subsection{Paper structure} 
The remainder of the paper is structured as follows: 
\begin{itemize} 
\item Section \ref{rw} describes related work. 
\item Section \ref{setup} describes the formal model, defines BA, $\epsilon$-BA, and Extractable BA, and establishes reductions between them. 
\item Section \ref{pksection} recalls Recursive Phase King (on which our protocol for $\epsilon$-BA is based). 
\item Section \ref{random} describes and analyses a randomised protocol for $\epsilon$-BA. 
\item Section \ref{sec:derandomization} describes how to derandomise the protocol of Section \ref{random}. 
\item Section \ref{discuss} discusses remaining open problems. 
\item Appendix \ref{AppA} observes that any protocol in which every processor receives a message has message complexity $\Omega(n+f^2)$. 
\item Appendix \ref{fbound} gives a simple proof that reaching univalency for BA requires message complexity $\Omega(f)$.

\end{itemize} 

\subsection{Related Work} \label{rw}

Below we refer to both \emph{message complexity}, which is the maximum number of messages sent by all correct processors (combined) in any execution, and \emph{communication complexity}, which counts instead the number of bits. Byzantine Agreement (BA) and Byzantine Broadcast (BB) are defined in Section \ref{setup}. In the synchronous setting, standard reductions (see, e.g., \cite{lynch1996distributed}) show that any protocol for BA can be used to solve BB with only $O(n)$ additional messages of constant size, so lower bounds for BB carry over to BA. The converse reduction is less direct: while BB can be used to solve BA, the overhead may be greater in some settings, meaning BA can require higher message or communication complexity than BB.

\vspace{0.2cm} 
\noindent \textbf{The results of Dolev and Reischuk.} Dolev and Reischuk~\cite{dolev1985bounds} showed that any protocol for BB\footnote{Dolev and Reischuk use the term ``Byzantine Agreement'' in~\cite{dolev1985bounds} to refer to the task now most commonly called ``Byzantine Broadcast'' (as defined in Section \ref{setup}).} has message complexity $\Omega(f^2+n)$. While they explicitly consider Byzantine faults, the proof is easily seen to hold for omission faults as well. They also showed how to transform any protocol solving BB or BA with message complexity $O(n^2)$ into a protocol for the authenticated setting with message complexity $O(f^2+n)$. Additionally, they established that any protocol solving BB requires communication complexity $\Omega(nf)$ in the authenticated setting, and message complexity $\Omega(nf)$ in the unauthenticated setting. To our knowledge, it remains open whether BB or BA  can be solved with communication complexity $O(n+f^2)$ for omission faults.

\vspace{0.2cm} 
\noindent \textbf{Achieving the quadratic bound.} The positive result of Dolev and Reischuk, achieving $O(f^2+n)$ message complexity, assumes a protocol for BB or BA with message complexity $O(n^2)$. For BB in the authenticated setting, this is achieved by a simple variant of the protocol of Dolev and Strong~\cite{dolev1983authenticated}. In the unauthenticated setting with $f < n/3$, Berman, Garay, and Perry~\cite{berman1992bit} used the Recursive Phase King protocol (as described in Section \ref{pksection}) to show that  BA (and hence BB) can be solved with communication complexity $O(n^2)$. Neither BA nor BB is solvable in the unauthenticated setting when $f \geq n/3$~\cite{fischer1986easy}. In the authenticated setting with $f < n(1/2 - \epsilon)$ for some constant $\epsilon > 0$, Momose and Ren~\cite{momose2020optimal} showed that BA can be solved with communication complexity $O(n^2)$. It is also straightforward to see that, when $f < n/2$, any protocol solving BA or BB with communication complexity $O(n^2)$ can be converted into a protocol with communication complexity $O(fn)$: first run the given protocol on a committee of size $O(f)$, have committee members send their outputs to all processors, and then decide by majority vote.

\vspace{0.2cm} 
\noindent \textbf{Crash faults.} Galil, Mayer, and Yung~\cite{galil1995resolving} showed that Broadcast can be solved for crash faults with message complexity $O(n)$. A simple proof that Broadcast can be solved for crash faults with communication  complexity $O(n)$ also appears in~\cite{lewispyeAQC}.

\vspace{0.2cm} 
\noindent \textbf{Almost-Everywhere Agreement.} Dwork, Peleg, Pinter, and Upfal~\cite{dwork1986fault} introduced \emph{almost-everywhere agreement} for bounded-degree networks, where some processors may be ``surrounded'' by Byzantine nodes and cannot be reached. While their notion allows some correct processors to output incorrectly, it differs from our notion of $\epsilon$-BA in several significant regards, including using a different notion of Validity. The analysis of \cite{dwork1986fault} also has different motivations: almost-everywhere agreement arose from topological constraints (some processors are unreachable), while $\epsilon$-BA relaxes agreement on complete networks. Their analysis is concerned with describing network topologies upon which different resiliences are obtainable, rather than achieving tight bounds on communication and message complexity in complete networks. Abraham and Stern~\cite{abraham2024dolev} recently extended Dolev-Reischuk-style lower bounds to Crusader Broadcast and to an ``all-but-$m$'' relaxation. For the latter, they show that any $(2/3+\epsilon)$-correct probabilistic protocol for all-but-$(f^c-1)$ Crusader Broadcast tolerating $f$ strongly adaptive Byzantine faults must send at least $\Omega(\epsilon n f^{1-c})$ messages in expectation. Their setting differs from ours in several respects: they analyse Crusader Broadcast (where $\bot$ outputs are permitted when the sender is faulty), their lower bounds require a strongly adaptive adversary capable of simulating honest parties, and they consider probabilistic protocols. In contrast, we provide deterministic upper bounds for $\epsilon$-BA against static adversaries.

\vspace{0.2cm} 
\noindent \textbf{Implicit Byzantine Agreement.} Kumar and Molla~\cite{kumar2024sublinear} recently defined \emph{implicit BA}, where only a subset of honest nodes must agree. They achieve $\tilde{O}(\sqrt{n})$ message complexity using a global coin.  Our work differs in providing deterministic protocols and in our focus on the univalency question.

\vspace{0.2cm} 
\noindent \textbf{Approximate Agreement.} A long line of work \cite{dolev1986reaching,mahaney1985inexact,fekete1986asymptotically,fekete1990asymptotically,fekete1987asynchronous,kieckhafer1994reaching,abraham2004optimal,mendes2013multidimensional} studies \emph{approximate agreement}  where processors with real-valued inputs must output values within $\epsilon$ of each other. This is a different relaxation: all processors must agree approximately in value space, rather than allowing some processors to disagree entirely.

\vspace{0.2cm} 
\noindent \textbf{Randomised protocols.} There is also extensive research \cite{abraham2019synchronous,abraham2018validated,cachin2001secure,feldman1988optimal,katz2006expected,abraham2019communication,chen2016algorand,king2011breaking} analysing how randomness can be used to achieve quadratic or even sub-quadratic expected communication complexity. Our contribution is demonstrating that \emph{deterministic} protocols can reach \emph{univalency} with subquadratic communication.

\vspace{0.2cm} 
\noindent \textbf{Extension protocols.} A number of papers \cite{cachin2005asynchronous,lu2020dumbo,miller2016honey,nayak2020improved} consider so-called \emph{extension} protocols, which look to optimise communication complexity when inputs lengths $\ell$ are large compared to $n$. When $n\ell$ dominates $n^2$, such protocols look to achieve optimal $O(n\ell)$ communication complexity, often using a BA oracle for short inputs. 

\vspace{0.2cm} 
\noindent \textbf{Extractable SMR.} The blockchain literature has extensively studied data availability sampling~\cite{al2021fraud}, where nodes verify data availability through random sampling without downloading entire blocks. The notion of Extractable State-Machine-Replication (Extractable SMR) \cite{lewis2025morpheus} was introduced to provide a  theoretical framework connecting these practical constructions to classical SMR. The notion of Extractable BA is the analogous notion for Byzantine Agreement.  

\section{Model and Definitions} \label{setup} 

\subsection{System Model}

\textbf{The set of processors}. We consider a set $\Pi=\{ p_1,\dots,p_n\} $ of $n$ processors, at most $f$ of which may be \emph{faulty}.\footnote{For simplicity of presentation, we initially consider a static adversary. As observed in Section \ref{discuss}, however, the fact that our protocol is deterministic means that our results carry over to adaptive adversaries.}  Processors that are not faulty are \emph{correct}. Faulty processors may display Byzantine (i.e., arbitrary) behaviour.  

\vspace{0.2cm} 
\noindent   \textbf{Time}. Time is divided into discrete time-slots $t=0,1,2,\dots$. 

\vspace{0.2cm} 
\noindent \textbf{Message delivery}. Processors communicate using  point-to-point channels.   We consider the standard synchronous (or `lock-step') model, in which a message sent at time-slot $t$ is received at time-slot $t+1$. Our results extend straightforwardly to a setting in which there is some known $\Delta$ such that a message sent at time-slot $t$ is received at some time-slot in $(t,t+\Delta]$ (carry out the instructions for time-slot $t$ at time-slot $t\cdot \Delta$). 

\vspace{0.2cm} 
\noindent \textbf{Processors are modelled as state-transition-diagrams}. For the sake of concreteness, we model processors as state-transition-diagrams. This means that, at each time-slot $t$, each processor $p_i$ begins in some \emph{state} $x$ and then receives a set of messages on each of its communication channels.  Given $x$ and a specification of which set of messages has arrived on each channel  at $t$,  the state-transition-diagram then determines: (a) The message (if any) that $p_i$ sends along each of its channels at $t$, and; (b) The state in which  $p_i$ begins time-slot $t+1$. So, according to this model, the state-transition-diagram for each processor specifies the protocol. There is no requirement that the state-transition-diagram be finite, i.e., it may contain an infinite number of states. Faulty processors can have arbitrary state-transition-diagrams. 

\vspace{0.2cm} 
\noindent \textbf{Inputs and outputs}. Certain states of the state-transition-diagram for $p_i$ are labelled as `input states', and $p_i$ begins time-slot 0 in one of these input states.  The input to $p_i$ determines which input state it starts in. The number of processors $n$, $f$,  and also $p_i$'s name `$p_i$' are given to $p_i$ as part of its input. Certain states are also labelled as `output' states. When processor $p_i$ first reaches such a state $x$ at time $t$, its output is determined by the state $x$.

\vspace{0.2cm} 
\noindent \textbf{The authenticated setting}. The information-theoretic model above, in which processors are modelled as state-transition diagrams (with no restriction on the set of `allowable' diagrams) describes the \emph{unauthenticated} setting. In the \emph{authenticated} setting, we assume that a signature scheme is available. This can be simply formalised using the approach in \cite{lewis2023permissionless}. If we work in the authenticated setting, this means that: 
\begin{itemize} 
\item For each $p_i\in \Pi$ and each message $m$, there exists a special message, denoted  $\langle m \rangle_{i}$, to be thought of as the `message $m$ signed by $p_i$'. 
\item For each $p_i\in \Pi$ and $m$, processors other than $p_i$ cannot send any message containing the sequence  $\langle m \rangle_{i}$ until any time-slot at which they have received a message containing this sequence. This means we restrict the set of allowable state-transition-diagrams (even for Byzantine processors) to those which obey this requirement.  
\end{itemize} 

\vspace{0.2cm} 
\noindent \textbf{Message complexity and communication complexity}. The message complexity of a protocol is the maximum number of messages sent by all correct processors (combined) in any execution. The communication complexity  of a protocol is the maximum number of bits sent by all correct processors in any execution.

\vspace{0.2cm} 
\noindent \textbf{Omission and crash faults}. While we are mainly concerned with Byzantine faults in this paper, Appendices \ref{AppA} and \ref{fbound} also consider crash and omission faults. In the crash-fault model,  faulty processors have the state-transition-diagram specified by the protocol and execute instructions correctly, except that they may \emph{crash} at some time-slot $t$. If $p_i$ crashes at $t$, then it may send an arbitrary subset of the messages that it is instructed to send at $t$, and $p_i$ then performs no further action at later time-slots (remaining in the state that it was in at  time-slot $t$ thereafter, while sending and receiving no further messages). 

In the omission-fault model,  faulty processors have the state-transition-diagram specified by the protocol and execute instructions correctly, except that they may fail to send or receive messages. Specifically, at each time-slot, a processor $p$ with omission-faults may take an arbitrary subset of the messages that have arrived on each channel, and will then act (at that time-slot) exactly like a correct processor that has received this (potentially reduced) set of messages, except that it may send an arbitrary subset of the messages that the instructions require it to send on each channel.

\subsection{Byzantine Agreement, Byzantine Broadcast and Univalency} 

For the sake of simplicity when analysing communication complexity, we consider \emph{binary} versions of BA and BB, where processors are given an input in $\{ 0, 1 \}$. However, our protocol works unchanged for the more general version of BA  in which processors are given inputs in some known finite set.

\vspace{0.2cm}  
\noindent \textbf{Byzantine Agreement}. Each processor in $\Pi$ is given an input in $\{ 0, 1 \}$. The  following conditions must be satisfied:
\begin{itemize} 
\item \textbf{Termination}. All correct processors  give an output in $\{ 0, 1 \}$. 

\item \textbf{Agreement}. No two correct processors give different outputs.

\item \textbf{Validity}. If all correct processors have the same input $v$, then $v$ is their common output. 
\end{itemize} 

\vspace{0.2cm}  
\noindent \textbf{Byzantine Broadcast}.  One processor, designated the `broadcaster', is given an input in $\{ 0, 1 \}$. All processors are given the name of the broadcaster. 
The following conditions must be satisfied: 
\begin{itemize} 
\item \textbf{Termination}. All correct processors give an output in  $\{ 0, 1 \}$. 

\item \textbf{Agreement}. No two correct processors  give different outputs.

\item \textbf{Validity}. If the broadcaster is correct and has input $v$, all correct processors output $v$.
\end{itemize}

\vspace{0.2cm}  
\noindent \textbf{Executions and  univalent runs}.  An \emph{execution} is a complete description of the inputs to each processor, and the messages sent by each processor at each time-slot. For each $k\in \mathbb{N}$, a \emph{$k$-run} is a complete description of the inputs to each processor, and the messages sent  by each processor at each time-slot $\leq k$. If $r$ is a  $k$-run, and if  all processors receive the same inputs and  send the same messages at time-slots $\leq k$ in $r$ as they do in an execution $E$, then we say the execution $E$ \emph{extends} $r$. A $k$-run $r$ is \emph{univalent} if there exists $v$ such that all correct processors output $v$ in all executions extending  $r$.

\subsection{$\epsilon$-Byzantine Agreement}

We define a version of BA that allows up to $n\epsilon$ correct processors to output `incorrectly'.

\vspace{0.2cm} 
\noindent \textbf{Defining $\epsilon$-BA}. 
For $\epsilon \in (0, 1/2)$, a protocol solves $\epsilon$-BA  if each processor receives an input in $\{0, 1\}$ and produces an output in $\{0, 1\}$, satisfying:
\begin{itemize}
    \item \textbf{Termination:} Every correct processor produces an output.
    \item \textbf{$\epsilon$-Agreement:} There exists a set $X$ of correct processors with $|X| < n\epsilon$ such that all correct processors outside $X$ output the same value $v$. (The protocol need not identify $X$.)
    \item \textbf{$\epsilon$-Validity:} If fewer than $n\epsilon$ correct processors have an input other than some value $x$, then $v = x$.
\end{itemize}

The crucial observation is that $\epsilon$-BA serves as the univalency phase of BA:

\begin{proposition}\label{prop:epsilon-to-ba}
Suppose $\epsilon\in (0,1/2)$ and  $f < n(1/2 - \epsilon)$. Any protocol for $\epsilon$-BA can be extended to a full BA protocol by adding a dissemination phase: each processor sends its $\epsilon$-BA output to all others, then outputs the majority value received.
\end{proposition}

\begin{proof}
After $\epsilon$-BA, at least $n-n\epsilon -f>n/2$ correct processors output the same value $v$. So $v$ is the majority value received by all correct processors, and all correct processors give $v$ as their BA output. Validity for $\epsilon$-BA ensures that if all correct processors are given the same input, then $v$ is equal to this common input. 
\end{proof}

\subsection{Extractable Byzantine Agreement}

\textbf{Motivation.} Many protocols for State Machine Replication (SMR)~\cite{schneider1990implementing} do not explicitly specify the mechanism by which correct processors retrieve the full finalised log of transactions. Instead, they ensure \emph{data availability} for the log. In Narwhal~\cite{danezis2022narwhal}, for example, blocks of transactions are sent to all processors, who then send acknowledgement messages confirming receipt. A set of $f+1$ such acknowledgements acts as a \emph{data availability certificate} for the block: it proves that at least one correct processor has received the block, and can therefore supply it to others upon request. Consensus can then be run on certificates rather than full blocks. 

The notion of \emph{Extractable SMR} was introduced in~\cite{lewis2025morpheus} to formalise the task being solved by such protocols. Roughly speaking, the requirement is that, while an individual processor may not be able to determine the full finalised log, the correct processors collectively receive sufficient messages to determine it. 

In this section, we introduce an analogous notion for BA.

\vspace{0.2cm} 
\noindent \textbf{Defining Extractable BA}. 
Write $\sigma \preceq \tau$ to denote that the string $\sigma$ is a prefix of the string $\tau$. 
For Extractable BA, we require the existence of a function $\mathcal{F}$ (to be specified by the protocol), which takes sets of (signed) messages as input and outputs a string, either 0, 1 or the empty string (of length 0). 
Let $M_c$ be all messages sent or received by correct processors over the protocol execution. Let $M_a$ be  all messages sent or received by any processor over the protocol execution. We require: 

\begin{itemize} 

\item \textbf{Consistency}: For all $M\subseteq M_a$: $\mathcal{F}(M)\preceq \mathcal{F}(M_a)$. 

\item \textbf{Termination}: $\mathcal{F}(M_c)\in \{ 0, 1 \}$. 

\item \textbf{Validity}: If all correct processors have the same input $x$, then $\mathcal{F}(M_c)=x$. 
\end{itemize} 

To understand this definition, note that, since  $\mathcal{F}(M_c)\in \{ 0, 1 \}$, it follows from Consistency that $\mathcal{F}(M)=\mathcal{F}(M_c)$ for \emph{any} set of messages $M$ with $M_c\subseteq M \subseteq M_a$. So the consistency condition ensures monotonicity: once a subset of messages ``proves'' a value, no superset can prove a different value.

\vspace{0.2cm} 
Let $M_c[t]$ be all messages sent or received by correct processors by time-slot $t$. We say a protocol for Extractable BA is \emph{finitely terminating} if there exists $t^*$ (depending on the protocol and $n$) such that Termination is always satisfied w.r.t.\ $t^*$, i.e., every execution satisfies the condition that $\mathcal{F}(M_c[t^*])\in \{ 0, 1 \}$.  As with $\epsilon$-BA, Extractable BA serves as the univalency phase of BA:

\begin{proposition}\label{prop:extractable-to-ba}
Any finitely terminating protocol for Extractable BA can be extended to a full BA protocol: after executing Extractable BA, each processor broadcasts all its messages to all others, then applies $\mathcal{F}$ to the collected messages.
\end{proposition}

\begin{proof}
After the dissemination phase, every correct processor holds a superset of $M_c$ (as defined for the execution of Extractable BA). By Termination for Extractable BA, $\mathcal{F}(M_c) \in \{0, 1\}$. By Consistency, all correct processors give the same output for BA, irrespective of which messages are disseminated by faulty processors. 
Validity for BA follows directly from the Validity condition of Extractable BA.
\end{proof}

\subsection{Reduction from Extractable BA to $\epsilon$-BA} 

\begin{proposition}\label{prop:reduction}
Suppose $\epsilon \in (0,1/3)$ and $f < n(1/3 - \epsilon)$.  If $\epsilon$-BA can be solved with communication complexity $C(n)$ (for some function $C$), then Extractable BA can be solved with communication complexity  $C(m) + m$, where $m = \lfloor f/(1/3 - \epsilon)\rfloor +1$.
\end{proposition}

\begin{proof}
Run $\epsilon$-BA on a committee of $m$ processors. This committee tolerates $f$ Byzantine faults since $f < m(1/3 - \epsilon)$. After $\epsilon$-BA, each committee member signs and sends their output to a single processor. Given any set of messages containing signed outputs from at least $m-f$ committee members, the extraction function $\mathcal{F}$ returns the majority signed value (breaking ties arbitrarily). Given any set of messages containing signed outputs from fewer than $m-f$ committee members, $\mathcal{F}$ outputs the empty string. 
\end{proof}

\section{Background: Recursive Phase King} \label{pksection} 
The protocol we present in Sections \ref{random} and \ref{sec:derandomization} is a modification of Recursive Phase King~\cite{berman1992bit}. As preparation, we briefly recall the key ideas here. Section \ref{bpk} describes (a slight modification of) the basic Phase King protocol, and Section \ref{rpk} describes its recursive variant.

\subsection{Basic Phase King} \label{bpk}

Phase King is a protocol for BA, and so must satisfy Termination, Agreement, and Validity. To explain it, we first consider a simple two-round sub-protocol called \emph{Gradecast}, which satisfies Termination and Validity, but not yet Agreement. 

 \vspace{0.2cm} 
 \noindent \textbf{Gradecast.} As in BA, each processor receives an input $v\in \{ 0, 1 \}$. Now, though, processors output a value in $\{ 0,1 \}$ together with a \emph{grade} in $\{ 0,1,2 \}$, which indicates something about the processor's knowledge regarding other processors' outputs. The protocol is shown in Figure \ref{f3}.

\begin{figure} 
 \vspace{0.2cm} 

\begin{tcolorbox}[colback=gray!5!white,colframe=black!75!black]
 \textbf{Gradecast: Instructions for $p_i$}.

\vspace{0.1cm} 

   \emph{Time-slot 0.} Send $v$ to all processors. 

\vspace{0.1cm} 

   \emph{Time-slot 1.} If there exists $b$ such that at least $n-f$ processors sent $b$ to $p_i$ at time-slot 0, send $b$ to all processors. 
                
                \vspace{0.1cm} 

   \emph{Time-slot 2.} Output as follows: 
   \begin{itemize} 
   \item[(i)]  If there exists $b$ such that $n-f$ processors sent $b$ to $p_i$ at time-slot 1, output $b$ with grade 2. 
    \item[(ii)] Otherwise, if there exists $b$ such that $f+1$ processors sent $b$ to $p_i$ at time-slot 1, output $b$ with grade 1.
     \item[(iii)] Otherwise, output $v$ with grade 0.
     \end{itemize} 
     \end{tcolorbox}
     
     \caption{Gradecast} 
     \label{f3} 
     \end{figure} 
    
     \vspace{0.2cm} 
 Let us observe some simple properties of this protocol. It follows from standard quorum intersection arguments that it cannot be the case that one processor $p_i$ is sent $b$ by $n-f$ processors at time-slot 0, while another processor $p_j$ (possibly $p_i=p_j$) is sent $b'\neq b$ by $n-f$ processors. To phrase this another way, the protocol satisfies \emph{Time-slot 1 Agreement}: 
 
 \begin{quote} 
\textbf{Time-slot 1 Agreement.} If $p_i$ is correct and sends $b$ at time-slot 1, then no correct processor sends $b'\neq b$ at time-slot 1. 
\end{quote}

 By the same argument, when $p_i$ outputs at time-slot 2, it cannot be the case that (i) holds for two distinct values of $b$. It follows directly from Time-slot 1 Agreement that (ii) cannot hold for two distinct values of $b$, since if $f+1$ processors send $b$ to $p_i$ at time-slot 1 then at least one is correct. This suffices to show that the protocol is well-defined: processors are instructed to send at most one value at each time-slot, and the instructions specify a unique output for each processor. 
 
 It is also easy to see that the protocol satisfies a  form of Validity: 
 \begin{quote} 
\textbf{Validity}$^{+}$. If all correct processors have the same input $v$, then they all output $v$ with grade 2. 
\end{quote} 

\noindent Outputting $v$ with grade 2 also implies knowledge about the outputs of other processors: 
\begin{quote} 
 \textbf{Knowledge of Agreement.} If any correct processor outputs a value $b$ with grade 2, then all correct processors output $b$. 
\end{quote} 

\noindent So if a correct processor outputs $b$ with grade 2, then it \emph{knows} all correct processors have output $b$. 
To prove Knowledge of Agreement, note that if a correct processor outputs $b$ with grade 2, then it received $b$ from at least $n-f$ processors at time-slot 1. This means all correct processors received $b$ from at least $n-2f\geq f+1$ processors at that time-slot. 

\vspace{0.2cm} 
\noindent  \textbf{Introducing tie-breaking.} So far, we have a simple protocol in which, if all correct processors start with the same input, they all output that input with grade 2.
What we have not achieved yet is any method for tie-breaking: if the processors start with a mix of inputs, how do we get them to agree on some value? 
 The basic idea is as follows: 
 \begin{itemize} 
 \item We run $f+1$ \emph{rounds}, each with a different \emph{leader}. 
 \item Processors start each round with a certain value---at the start of the first round, this is simply their input. 
 \item Each round begins with an instance of Gradecast, after which processors update their value to be that specified by their output in the Gradecast instance.
 \item After Gradecast, the leader of the round shares their value $v$. Each correct processor changes their value to $v$ \emph{unless} they output the Gradecast instance for this round with grade 2, before progressing to the next round (if this is not the last). 
 \item At the end of the last round, processors output their current value.  
 \end{itemize} 
 
  If we proceed in this way, then Termination will be satisfied. It is also straightforward to see that Validity will be satisfied: if all correct processors start with the same value $v$, then it follows inductively (from Validity$^+$) that they will output $v$ with grade 2 in each Gradecast instance, and will not change their value upon hearing from the leader. 
 
 It is also straightforward to see that, after any round with a correct leader (and in all subsequent rounds), all correct processors will have the same value: if a correct processor does not change their value to that of the correct leader because they output with grade 2 in the Gradecast instance, then Knowledge of Agreement means that the leader must anyway have the same value. Since at least one of the $f+1$ leaders is correct, Agreement is satisfied.

\subsection{Recursive Phase King} \label{rpk}

Phase King requires $f+1$ rounds because we need to ensure there is at least one round with a correct leader. The first such round suffices to ensure all correct processors have the same value, and once this is the case, agreement is maintained in subsequent rounds.

 \vspace{0.2cm} 
 \noindent \textbf{Replacing $f+1$ rounds with two.} The key observation behind Recursive Phase King is that we can replace the $f+1$ rounds with two if we replace leaders with \emph{committees}. We partition the $n$ processors into two committees, $C_1$ and $C_2$, which differ in size by at most 1. At least one of these committees must be \emph{sufficiently honest}, meaning it has a fraction of faulty processors at most $f/n$. Rather than having a single leader send their value in round 1, we have committee $C_1$ carry out a recursive call to BA, and then send their outputs to all processors. If $C_1$ is sufficiently honest, the committee members will agree on a value, and all processors will then receive the same majority value from the committee. In round 2, we do the same using committee $C_2$. Since at least one of $C_1$ or $C_2$ is sufficiently honest, at least one of them effectively acts as a correct leader. 
 
  \vspace{0.2cm} 
 \noindent \textbf{Looking forward to Sections \ref{random} and \ref{sec:derandomization}.} The aim of Recursive Phase King is to achieve communication complexity $O(n^2)$. For some constant $c\in \mathbb{N}$, the all-to-all communication in rounds 1 and 2 requires correct processors to send at most $cn^2$ bits (at `depth 1', say) \emph{in addition to} those bits required by the two recursive calls. Ignoring rounding errors, those two recursive calls require the sending of $c(n/2)^2 +c(n/2)^2=cn^2/2$ bits (at `depth 2', say), and so on. Summing over all depths, the protocol requires correct processors to send at most $2cn^2$ bits. 
 
 In Section \ref{random}, we describe a probabilistic version of Recursive Phase King that requires correct processors to send only $O(n)$ messages at each depth. Then, in Section \ref{sec:derandomization}, we describe how to derandomise the protocol. Sending $O(n)$ messages at each of $O(\log n)$ depths yields communication complexity $O(n \log n)$.

\section{A Randomised Protocol for $\epsilon$-BA} \label{random} 

\subsection{Protocol specification} 
We first present a randomised protocol, then derandomise it in Section~\ref{sec:derandomization}. As in Recursive Phase King, we divide the $n$ processors into two committees, $C_1$ and $C_2$, differing in size by at most 1. The protocol invokes an instance of $\epsilon$-BA on each committee. In this section's analysis, it does not matter exactly how the committees carry out $\epsilon$-BA, so long as they do so correctly; for simplicity, one can imagine they use some deterministic protocol.\footnote{We require that all correct processors terminate and output a value in $\{0, 1\}$ within some known finite time bound, even when a majority of processors are Byzantine. This condition is also satisfied by the protocol we define to prove Theorem~\ref{thm:epsilon-ba-intro}.} 

\vspace{0.2cm} 
To reduce communication complexity, each processor samples values from $k$ randomly chosen processors, rather than sending values to all others. Here $k$ is a fixed constant independent of $n$. The basic idea is that, so long as $k$ is sufficiently large, most correct processors should output correctly. In Section~\ref{sec:derandomization}, we show via a counting argument that there exists a \emph{fixed} specification of the sampling choices $\mathtt{Smp}$ that causes most correct processors to output correctly. This holds regardless of the processors' inputs, the choice of which processors are faulty, and the behaviour of faulty processors.

\vspace{0.2cm} 
The probabilistic protocol, which we call Probabilistic $\epsilon$-RPK (where RPK stands for Recursive Phase King), is described in Figure \ref{fig1}. Each processor maintains a value $v \in \{0, 1\}$, initialised to its input. 

\begin{figure}
 \vspace{0.2cm} 

\begin{tcolorbox}[colback=gray!5!white,colframe=black!75!black]
 \textbf{Probabilistic $\epsilon$-RPK: Instructions for round $r\in \{ 1,2 \}$}.

\vspace{0.1cm} 

\noindent \emph{Step 1.} Each processor $p$ samples $k$ processors uniformly at random (with replacement) and learns their values. If at least $k(2/3 - \epsilon/2)$ report the same value $b$, processor $p$ sets its \emph{response value} to $b$; otherwise, $p$'s response value is $\bot$.

\vspace{0.1cm} 
\noindent \emph{Step 2.} 
Each processor $p$ independently samples $k$ processors and learns their response values. Based on responses:
\begin{itemize}
    \item If $\geq k \cdot 2/3$ responses equal some $b \in \{0,1\}$: set $v := b$ with grade 2.
    \item Else if there exists a unique $b \in \{0,1\}$ with $\geq k/3$ responses equal to $b$: set $v := b$ with grade 1.
    \item Else: keep $v$ unchanged with grade 0.
\end{itemize}

\vspace{0.1cm} 
\noindent \emph{Step 3.} 
Run $\epsilon$-BA on $C_r$. Then each processor $p$ samples $k$ members of $C_r$ to find their output:
\begin{itemize}
    \item If $p$ has grade 2: keep current value.
    \item Otherwise: set $v$ to the majority response from $C_r$, breaking ties arbitrarily.
\end{itemize}
If $r=1$, proceed to round 2. Otherwise, output $v$.

     \end{tcolorbox}

\caption{Probabilistic $\epsilon$-RPK}
 \label{fig1} 
\end{figure} 

\subsection{Analysis}

Suppose $n>2$, $\epsilon \in (0,1/3)$, and $f<n(1/3-\epsilon)$. In this section, we define what it means for a sampling choice to be \emph{bad}, then show that if the sampling choice is not bad, Probabilistic $\epsilon$-RPK achieves $\epsilon$-BA. 

\vspace{0.2cm} 
\noindent \textbf{Defining bad sampling choices.} The protocol consists of six steps (three per round). A \emph{sampling choice} is a specification of which $k$ processors each correct processor samples in each of the six steps. We introduce the following terminology: 
\begin{itemize} 
\item For steps 1, 2, 4, and 5, the \emph{response set} is the set of all processors $\Pi$ (these are the processors that may be queried).
\item For step 3, the response set is $C_1$; for step 6, the response set is $C_2$. 
\item In each step, a correct processor's \emph{response value} is the value it returns if queried.  
\end{itemize} 
\noindent For each step $s$ and $b\in \{ 0,1 \}$, let $\alpha_{s,b}$ be the fraction of the response set that are correct processors with response value $b$. For a correct processor $p$, let $\alpha_{s,b,p}$ be the fraction of the $k$ processors sampled by $p$ in step $s$ that are correct and have response value $b$. 

A sampling choice is \emph{bad for $p$} if there exist $s\in \{1,\ldots,6\}$ and $b\in \{0, 1\}$ such that $|\alpha_{s,b}-\alpha_{s,b,p}|\geq \epsilon/2$. 
A sampling choice is \emph{bad} if it is bad for at least $n\epsilon/4$ correct processors.

\vspace{0.2cm} 
We say \emph{sampling is not bad} if the sampling choice selected during a given execution is not bad.

\begin{lemma}\label{lem:good-implies-correct}
If sampling is not bad, Probabilistic $\epsilon$-RPK achieves $\epsilon$-BA.
\end{lemma}

\begin{proof} Suppose sampling is not bad. Termination is immediate, so it remains to prove $\epsilon$-Validity and $\epsilon$-Agreement. 

\vspace{0.1cm} 
\noindent \emph{$\epsilon$-Validity.} Suppose fewer than $n\epsilon$ correct processors have input other than $b$. Since $f<n(1/3-\epsilon)$, at least $2/3$ of all processors are correct with input $b$. If $p$ is correct and the sampling choice is not bad for $p$, then $p$ sets its response value to $b$ in step 1. Since sampling is not bad, at least $n(2/3+\epsilon/2)$ processors are correct and set their response value to $b$ in step 1. All correct processors for whom the sampling choice is not bad then set their value to $b$ with grade 2 in step 2, and do not change their value in step 3. The same argument applies to steps 4--6. 

\vspace{0.1cm} 
\noindent \emph{$\epsilon$-Agreement.} Suppose the first committee with Byzantine fraction at most $f/n$ is $C_1$ (a similar argument applies if it is $C_2$). If at least $n\epsilon/4$ correct processors set their value to some $b \in \{0,1\}$ with grade 2 in step 2, then this includes some processor for whom the sampling choice is not bad. It follows that $\alpha_{2,b}>2/3 -\epsilon/2$, so $b$ is unique and all correct processors for whom the sampling choice is not bad set their value to $b$ with grade at least 1 in step 2. The number of processors in $C_1$ that are correct and do not use $b$ as their input to the $\epsilon$-BA instance on $C_1$ is therefore less than $n\epsilon/4$. Since $|C_1|\geq \lfloor n/2 \rfloor$ and $n>2$, this is less than $|C_1| \cdot \epsilon$. By $\epsilon$-Validity, correct processors in $C_1$ output $b$ from the $\epsilon$-BA instance on $C_1$. 

Whether or not at least $n\epsilon/4$ correct processors set their value to some $b$ with grade 2 in step 2, all correct processors for whom the sampling choice is not bad set their value to the committee's output $v$ in step 3 (or already have value $v$ and keep it). Applying the same argument as for $\epsilon$-Validity to steps 4--6, all correct processors for whom the sampling choice is not bad output $v$. 
\end{proof}

\section{Derandomisation}\label{sec:derandomization}

We now show that the randomised protocol can be derandomised. Specifically, for sufficiently large $k$, there exists a \emph{fixed} sampling choice $\mathtt{Smp}_n$ for each $n$ that is not bad, regardless of the processors' inputs, the identity of the faulty processors, or their behaviour. 

\vspace{0.2cm}
\noindent \textbf{Breaking sampling choices into steps.} 
A \emph{step $s$ sampling choice} specifies which $k$ processors each processor samples in step $s$. A  \emph{step $s$ configuration} specifies, for each processor in the corresponding response set, whether it is Byzantine or correct, and, in the latter case:
\begin{itemize}
    \item For $s \in \{1, 3, 4, 6\}$, whether it has response value 0 or 1.
    \item For $s \in \{2, 5\}$, whether it has response value 0, 1, or $\bot$.
\end{itemize}

\noindent For each $b \in \{0,1\}$ and step $s$ configuration $z$, let $\alpha_{s,b}(z)$ be the fraction of the response set that are correct processors with response value $b$ in $z$. For a correct processor $p$ and step $s$ sampling choice $\mathtt{Smp}_n^s$, let $\alpha_{s,b,p}(z, \mathtt{Smp}_n^s)$ be the fraction of the $k$ processors sampled by $p$ according to $\mathtt{Smp}_n^s$ that are correct with response value $b$ in $z$. 

We say $\mathtt{Smp}_n^s$ is \emph{bad for $p$ given $z$} if $|\alpha_{s,b}(z) - \alpha_{s,b,p}(z, \mathtt{Smp}_n^s)| \geq \epsilon/2$ for some $b \in \{0,1\}$. We say $\mathtt{Smp}_n^s$ is \emph{bad given $z$} if it is bad for at least $n\epsilon/24$ correct processors given $z$.\footnote{The factor of 24 arises from six steps and the threshold of $n\epsilon/4$ in the definition of a bad sampling choice.} We say $\mathtt{Smp}_n^s$ is \emph{bad} if it is bad given some step $s$ configuration.

\vspace{0.2cm}
\noindent \textbf{Step 1.}
The number of step 1 configurations is at most $3^n$ (each processor is Byzantine, or correct with response value 0, or correct with response value 1). We show that for sufficiently large $k$, some step 1 sampling choice is not bad.

Fix a configuration $z$ and consider the uniform distribution on step 1 sampling choices. Under this distribution, each processor independently samples $k$ processors uniformly at random with replacement. Let $X_i$ be the indicator that the sampling choice is bad for processor $p_i$. By the Chernoff bound, the probability that $p_i$'s sample deviates from the population fraction by at least $\epsilon/2$ can be made arbitrarily small by choosing $k$ large. Specifically, for any $\alpha > 0$, there exists $k = k(\alpha, \epsilon)$ such that $\Pr[X_i = 1] \leq \alpha$ for all $i$.

Since each processor samples independently, $X_1, \ldots, X_n$ are independent random variables. Let $S_n = \sum_{i=1}^n X_i$. We apply the Chernoff bound with KL divergence (see, e.g., \cite{mitzenmacher2005probability}): for $q > p$,
\[
\Pr[S_n \geq nq] \leq e^{-n D(q \| p)},
\]
where $D(q \| p) = q \ln(q/p) + (1-q) \ln((1-q)/(1-p))$ is the Kullback--Leibler divergence.

Setting $p = \alpha$ and $q = \epsilon/24$ (with $\alpha < \epsilon/24$), we require $D(\epsilon/24 \| \alpha) > \ln 3$ to ensure $\Pr[S_n \geq n\epsilon/24] < 3^{-n}$. As $\alpha \to 0$:
\[
D(\epsilon/24 \| \alpha) = \frac{\epsilon}{24} \ln\frac{\epsilon/24}{\alpha} + \left(1 - \frac{\epsilon}{24}\right) \ln\frac{1 - \epsilon/24}{1 - \alpha} \to \infty,
\]
since the first term diverges while the second approaches $(1 - \epsilon/24) \ln(1 - \epsilon/24)$. Thus, for sufficiently small $\alpha$ (achieved by sufficiently large $k$), we have $D(\epsilon/24 \| \alpha) > \ln 3$.
By a union bound over all at most $3^n$ configurations, the probability that a uniformly random sampling choice is bad is then less than $3^n \cdot 3^{-n} = 1$. Hence some step 1 sampling choice is not bad.

\vspace{0.2cm}
\noindent \textbf{Steps 2--6.}
The analysis for steps 2 and 5 is identical, except that configurations may also assign response value $\bot$, giving at most $4^n$ configurations. The required bound $D(\epsilon/24 \| \alpha) > \ln 4$ is achieved by the same reasoning. For steps 3 and 6, the response set is $C_r$ rather than $\Pi$, but the number of configurations remains less than $3^n$, so the same analysis applies.

Combining the six steps, for sufficiently large $k$, there exists a sampling choice $\mathtt{Smp}_n$ that is not bad in any execution.

\vspace{0.2cm}
\noindent \textbf{Specifying the protocol.}
For each $n>2$, fix a sampling choice $\mathtt{Smp}_n$ that is not bad in any execution. Fix also a protocol $P$ that solves $\epsilon$-BA for $n\leq 2$ with communication complexity at most some constant $C$ (such a protocol exists for any $\epsilon \in (0, 1/3)$ and $f < n(1/3 - \epsilon)$).

The instruction for $p$ to ``send values as specified by $\mathtt{Smp}_n$'' means $p$ sends its value to every processor that samples $p$ in that step according to $\mathtt{Smp}_n$. Similarly for response values. In each step, processors ignore messages from processors not required to send them messages by the protocol. The protocol is described in Figure \ref{fig2}.

\begin{figure} 
\vspace{0.2cm}

\begin{tcolorbox}[colback=gray!5!white,colframe=black!75!black]
\textbf{$\epsilon$-RPK: Instructions for round $r \in \{1, 2\}$}.

\vspace{0.1cm}

\noindent \emph{Step 1.} Send values as specified by $\mathtt{Smp}_n$. If at least $k(2/3 - \epsilon/2)$ received values equal some $b$, set response value to $b$; otherwise, set response value to $\bot$.

\vspace{0.1cm}
\noindent \emph{Step 2.} Send response values as specified by $\mathtt{Smp}_n$. Based on received response values:
\begin{itemize}
    \item If $\geq k \cdot 2/3$ equal some $b \in \{0,1\}$: set $v := b$ with grade 2.
    \item Else if a unique $b \in \{0,1\}$ has $\geq k/3$: set $v := b$ with grade 1.
    \item Else: keep $v$ unchanged with grade 0.
\end{itemize}

\vspace{0.1cm}
\noindent \emph{Step 3.} Run $P$ on $C_r$ if $|C_r| \leq 2$, or $\epsilon$-RPK on $C_r$ otherwise, with each processor using its current value as input. Processors in $C_r$ send their output to all processors that sample them according to $\mathtt{Smp}_n$. Each processor $p$ then updates:
\begin{itemize}
    \item If $p$ has grade 2: keep current value.
    \item Otherwise: set $v$ to the majority reported output from $C_r$ (ties broken arbitrarily).
\end{itemize}
If $r = 1$, proceed to round 2. Otherwise, output $v$.
\end{tcolorbox}
\caption{$\epsilon$-RPK}
\label{fig2} 
\end{figure} 

\vspace{0.2cm}
\noindent \textbf{Complexity analysis.}
The protocol has $O(\log n)$ levels of recursion, as committees halve in size at each level. At each level, the total number of processors across all committees is $n$, and the total number of bits sent by correct processors is $O(kn)$, where $k = k(\epsilon)$ depends only on $\epsilon$. There are $O(n)$ invocations of $P$ at the base of the recursion, each requiring at most $C$ bits. Total communication is therefore $O(n \log n)$.

\vspace{0.2cm}
\noindent \textbf{Completing the proofs of Theorems~\ref{thm:epsilon-ba-intro} and~\ref{thm:extractable-ba-intro}.}
That $\epsilon$-RPK solves $\epsilon$-BA follows from Lemma~\ref{lem:good-implies-correct} and our choice of $\mathtt{Smp}_n$. The complexity analysis completes the proof of Theorem~\ref{thm:epsilon-ba-intro}. Theorem~\ref{thm:extractable-ba-intro} then follows from Proposition~\ref{prop:reduction}.

\section{Discussion} \label{discuss} 

We have shown that reaching univalency for BA can be achieved with communication complexity $O(f \log f)$ when $f < n(1/3 - \epsilon)$. This precisely delineates what the Dolev-Reischuk bound measures: the $\Omega(f^2)$ lower bound reflects the cost of ensuring every correct processor can output---that is, dissemination---while reaching univalency requires only $O(f \log f)$ communication.

This decomposition has conceptual value: it reveals that ``reaching agreement'' and ``learning the agreement'' are fundamentally different computational tasks with different costs. The former is nearly linear in $f$; the latter is quadratic.

\vspace{0.2cm}
\noindent \textbf{Adaptive adversaries.} Although we explicitly considered a static adversary for the sake of simplicity, the result carries over immediately to adaptive adversaries: the protocol described in Section~\ref{sec:derandomization} is fully deterministic, so the adversary gains no advantage from observing protocol execution before choosing which processors to corrupt.

\vspace{0.2cm}
\noindent \textbf{Lower bounds.} For completeness, in Appendix~\ref{fbound} we give a simple proof that reaching univalency has message complexity $\Omega(f)$. Combined with our upper bound, this leaves a gap between $\Omega(f)$ and $O(f \log f)$.

\vspace{0.2cm}
\noindent \textbf{Open questions.} Our results leave several natural questions open:
\begin{itemize}
    \item For $f<n/3$ in the unauthenticated setting, is the communication complexity of reaching univalency for BA $\Omega(f \log f)$? More generally, what is the tight complexity as a function of both $n$ and $f$?
    \item What is the communication complexity of reaching univalency in the authenticated setting when $f \in (n/3, n/2)$? 
        \item Can similar separations between univalency and dissemination be established in asynchronous or partially synchronous models?
\end{itemize}

\section{Acknowledgements} 
We would like to thank Ittai Abraham and Ling Ren for a number of helpful conversations on the topic of this paper.

\bibliographystyle{ACM-Reference-Format}

\appendix

\section{The Dissemination Lower Bound}\label{AppA}

\begin{theorem} \label{DRmempool} 
Consider the authenticated setting with omission faults. Suppose $f<n-1$. Any deterministic protocol satisfying the condition that every correct processor receives at least one message   in every execution has message complexity   $\Omega(n+f^2)$. 
\end{theorem} 
\begin{proof} The proof is just a slightly simplified version of the proof of Dolev and Reischuk. 

That the message complexity is $\Omega(n)$ is immediate, so it remains to establish that the message complexity is  $\Omega(f^2)$. 
This is  trivially true if $f=0$, so we can suppose $f>0$.

Consider an arbitrary protocol.  First, we take an arbitrary subset $P_1\subseteq \Pi$ of size $\lceil f/2 \rceil$, setting $P_2=\Pi \setminus P_1$. Then: 
\begin{itemize} 
\item We consider the execution $E_1$,  in which the processors in $P_1$ are faulty. Processors in $P_1$ act correctly except that they do not send or receive messages to or from each other, and ignore messages from processors in $P_2$. 
\item If each processor in $P_1$ receives at least $\lceil f/2 \rceil$ messages from processors in $P_2$ in $E_1$, then we are done, since the correct processors in $P_2$ must then send at least $|P_1|\cdot f/2 \geq  (f/2)^2$ messages (combined). Otherwise, there exists $p\in P_1$ that receives at most $\lfloor f/2 \rfloor$ messages from processors in $P_2$. 
\item Now we use the fact that, because $P_1$ is of size  $\lceil f/2 \rceil$, we have $\lfloor f/2 \rfloor$ processors left to be faulty.  We consider an execution $E_2$ in which:
\begin{itemize} 
\item[-] Processor $p$ is correct. 
\item[-]  The remaining processors in $P_1$ act as in $E_1$, and all processors in $P_2$ act correctly except that they do not send messages to $p$ (this requires at most  $\lfloor f/2 \rfloor$ of them to be faulty). 
\end{itemize} 
\end{itemize} 
\emph{Analysis}. Since $p$ ignored messages from $P_2$ in $E_1$, does not receive messages from  $P_2$ in $E_2$, and does not receive messages from $P_1$ in either execution, it sends the same messages to processors in $P_2$ in both executions. In fact, $E_1$ and $E_2$ are indistinguishable for all processors in $P_2$ that are correct in $E_2$, and since $f<n-1$ this set is non-empty. So, $E_2$ is an execution in which  at most $f$ processors are faulty, and in which $p$ is a correct processor that does not receive any messages.  \end{proof}

\section{Reaching Univalency Requires $\Omega(f)$ Messages} \label{fbound}

In this section we consider crash faults and show that any protocol achieving univalency for BA has message complexity $\Omega(f)$.

\vspace{0.2cm}
\noindent \textbf{The setup.} We consider the state-transition-diagram model of Section~\ref{setup}, and suppose we have two protocols $P_1$ (for reaching univalency) and $P_2$ (for dissemination). We suppose that for some $t^*$, which may depend on $n$ and $f$, executing $P_1$ causes each correct processor to enter an ``intermediate output state'' at time $t^*$.\footnote{The time $t^*$ may also depend on processor inputs, but this is not important since there are only finitely many possible inputs for any given $n$.} Protocol $P_2$ is then executed, with each processor's input state for $P_2$ being its intermediate output state after executing $P_1$. We require that $P_2$ causes all correct processors to produce a BA output. We say that $P_1$ \emph{achieves univalency} if, for any execution of $P_1$, there exists a fixed $v \in \{0, 1\}$ such that any subsequent execution of $P_2$ causes all correct processors to output $v$.

\vspace{0.2cm}
\noindent \textbf{The proof.} Towards a contradiction, suppose there exist $P_1$, $P_2$, $n > f$, and $f > 0$, such that in every execution of $P_1$, fewer than $f/2$ correct processors send or receive a message.

Consider the two executions of $P_1$ in which all processors are correct and either all processors have input 0 or all processors have input 1. Let $A$ be the set of processors that send or receive a message in at least one of these two executions. By assumption, $|A| < f$.

Now consider $\mathcal{E}$, the set of all executions of $P_1$ in which processors in $A$ crash at time-slot 0, while processors outside $A$ are correct (with arbitrary inputs). Since processors outside $A$ neither send nor receive messages in either of the two executions above, and since crashed processors send no messages, processors outside $A$ send and receive no messages in any execution in $\mathcal{E}$.

\vspace{0.2cm} 
To complete the proof, we adapt the standard bivalency argument~\cite{fischer1985impossibility} to intermediate output states. Reorder processors so that $p_1, \ldots, p_k$ are the processors outside $A$, where $k = n - |A| > n - f$. For each $i \in \{0, \ldots, k\}$, let $E_i$ be the execution of $P_1$ in $\mathcal{E}$ in which each $p_j$ for $j \leq i$ has input 1, while each $p_j$ for $j > i$ has input 0. By validity, $E_0$ must be 0-valent and $E_k$ must be 1-valent. Hence there exists $i \in \{0, \ldots, k-1\}$ such that $E_i$ is 0-valent and $E_{i+1}$ is 1-valent.

Since $|A| < f$, there exist executions of $P_2$ extending $E_i$ and $E_{i+1}$ respectively, in which $p_{i+1}$ crashes immediately after $P_1$ completes, while all processors outside $A \cup \{p_{i+1}\}$ remain correct (this respects the fault bound since $|A \cup \{p_{i+1}\}| \leq f$). In these two complete executions (of $P_1$ followed by $P_2$), all correct processors have identical views: they receive no messages during $P_1$, and during $P_2$ they cannot distinguish the two executions since only $p_{i+1}$'s input differs, and $p_{i+1}$ has crashed. Yet one execution must output 0 and the other must output 1, giving the required contradiction.

\end{document}